\documentclass[pra,superscriptaddress,twocolumn]{revtex4-1}

\usepackage{verbatim}
\usepackage[utf8]{inputenc}
\usepackage[T1]{fontenc}
\usepackage[english]{babel}
\usepackage{ dsfont }

\usepackage{amsmath}
\usepackage{amssymb}
\usepackage{mathrsfs}
\usepackage{amsthm}
\usepackage{mathtools}
\usepackage{ wasysym } 
\usepackage{algorithm}
\usepackage{algpseudocode}
\usepackage{subcaption}

\usepackage{makecell}

\usepackage{stfloats} 

\usepackage{graphicx}
\usepackage{tikz}

\usepackage{hyperref}

\newcommand{\de}[1]{\left( #1 \right)}

\newcommand{\DE}[1]{\left\{#1\right\}}
\newcommand{\ket}[1]{\left| #1 \right\rangle}
\newcommand{\bra}[1]{\left\langle #1 \right|}
\newcommand{\ketbra}[2]{\left|#1\middle\rangle\middle\langle#2\right|}

\usepackage{sidecap}

\definecolor{darkred}{RGB}{200, 0, 0}
\definecolor{darkblue}{RGB}{0, 0, 180}
\definecolor{darkgreen}{RGB}{50, 150, 0}
\definecolor{col2}{HTML}{64A857}
\definecolor{col3}{HTML}{D1603D}
\definecolor{orange}{rgb}{1,0.5,0}

\newcommand*{\extra}[1]{}
\newcommand{\qx}{Q_X}
\newcommand{\qy}{Q_Y}
\newcommand{\qz}{Q_Z}

\newtheorem{thm}{Theorem}
\newtheorem{cor}{Corollary}
\newtheorem{prop}{Proposition}

\newtheorem{fact}{Observation}

\begin{document}

\title{Key rates for quantum key distribution protocols with asymmetric noise}

\author{Gl\'{a}ucia Murta}
\affiliation{QuTech, Delft University of Technology, Lorentzweg 1, 2628 CJ Delft, The Netherlands}
\affiliation{Institut f\"ur Theoretische Physik III, Heinrich-Heine-Universit\"at D\"usseldorf, Universit\"atsstraße 1, D-40225 D\"usseldorf, Germany}

\author{Filip Rozp\k{e}dek}
\affiliation{QuTech, Delft University of Technology, Lorentzweg 1, 2628 CJ Delft, The Netherlands}
\affiliation{Kavli Institute of Nanoscience, Delft University of Technology, Lorentzweg 1, 2628 CJ Delft, The Netherlands}
\author{J\'{e}r\'{e}my Ribeiro}
\affiliation{QuTech, Delft University of Technology, Lorentzweg 1, 2628 CJ Delft, The Netherlands}
\affiliation{Kavli Institute of Nanoscience, Delft University of Technology, Lorentzweg 1, 2628 CJ Delft, The Netherlands}
\author{David Elkouss}
\affiliation{QuTech, Delft University of Technology, Lorentzweg 1, 2628 CJ Delft, The Netherlands}
\author{Stephanie Wehner }
\affiliation{QuTech, Delft University of Technology, Lorentzweg 1, 2628 CJ Delft, The Netherlands}
\affiliation{Kavli Institute of Nanoscience, Delft University of Technology, Lorentzweg 1, 2628 CJ Delft, The Netherlands}

\begin{abstract}
We consider the asymptotic key rates achieved in the simplest quantum key distribution protocols, namely the BB84 and the six-state protocols, when non-uniform noise is present in the system.
We first observe that higher qubit error rates do not necessarily imply lower key rates. Secondly, we consider protocols with advantage distillation and show that it can be advantageous to use the basis with higher quantum bit error rate for the key generation. We then discuss the relation between advantage distillation and entanglement distillation protocols. 
We show that applying advantage distillation to a string of bits formed by the outcomes of measurements in the basis with higher quantum bit error rate is closely connected to the two-to-one entanglement distillation protocol DEJMPS \cite{DEJMPS}. Finally, we discuss the implications of these results for implementations of quantum key distribution.
\end{abstract}

\maketitle

\section{Introduction} 

Quantum key distribution (QKD) \cite{BB84,E91} is one of the most remarkable examples of the power of quantum mechanics. Many classical crypto-systems used for secure communication nowadays are based on computational assumptions. Computational assumptions make these systems vulnerable to retroactive attacks in case more powerful quantum computers become available in the future. In contrast, quantum mechanics allows two parties to distribute a key achieving information-theoretic security. This means that security is guaranteed even against an eavesdropper that has unlimited classical and quantum resources. Secure communication can then be achieved if this key is used in a one-time pad scheme~\cite{Vernam26,Shannon49} (for a discussion about the assumptions present in a QKD implementation, see \cite{LCT14}). 

Near term quantum technologies suffer from imperfections. Therefore, even in the absence of an active eavesdropper, a QKD implementation will be subjected to a finite amount of noise. In order to guarantee security, one is interested in designing protocols that can tolerate levels of noise compatible with current technology and at the same time achieve the highest possible key rates.

The simplest proposed QKD protocol, BB84~\cite{BB84}, is based on the conjugate coding ideas developed by Wiesner in~\cite{Wiesner83}. In the BB84, Alice prepares a single qubit state in one of the eigenstates of the Z-basis, $\DE{\ket{0},\ket{1}}$, or one of the eigenstates of the X-basis, $\DE{\ket{+},\ket{-}}$. An extension of the BB84, exploring three conjugate bases, was proposed in~\cite{6state} and is called the six-state protocol. In the six-state protocol, Alice can also prepare the qubit in one of the eigenstates of the Y-basis, $\DE{\ket{+_Y}=\frac{1}{\sqrt{2}}(\ket{0}+i\ket{1}),\ket{-_Y}\frac{1}{\sqrt{2}}(\ket{0}-i\ket{1})}$. The six-state protocol was proven to be more robust to noise~\cite{6state}. Intuitively, this is due to the fact that more parameters are characterized during the protocol, therefore restricting the possible actions of a potential eavesdropper.

In this article, we consider the asymptotic key rates that can be obtained in the BB84 and six-state protocols as a function of the quantum bit error rates (QBERs). We first consider instances of these protocols in which information reconciliation and privacy amplification are applied directly to the raw keys formed by the outcomes of Alice's and Bob's measurements.
Then, we consider the case when, additionally,
a sub-routine that allows Alice and Bob to select more correlated parts of their string, namely advantage distillation, is applied to the raw keys before information reconciliation. We discuss observed counter-intuitive behaviors of the asymptotic key rates with the QBER. As our main result, we show that in the presence of asymmetric noise, higher key rates may be obtained if the basis with higher QBER is used for the key generation in protocols with advantage distillation. This can have a direct impact for implementations that make use of advantage distillation{~\cite{rozpkedek2018near}}. Finally, we show that implementing the six-state protocol with advantage distillation and measurements in the basis with higher QBER, is closely connected to the two-to-one entanglement distillation protocol DEJMPS~\cite{DEJMPS}.

The manuscript is organized as following: in the remainder of this section we detail the general structure of the QKD protocols under consideration.
In Section~\ref{sec:results},  
we first consider the asymptotic key rates of the BB84 and the six-state protocols without advantage distillation. We then proceed to analyze the effect of advantage distillation and show interesting behaviors of the key rates as a function of the QBERs.
In Section~\ref{sec:entangdistil}, we discuss the relation of QKD and entanglement distillation protocols. {Finally, in Section~\ref{sec:experiments}, we discuss the implications of our results to experimental implementations.}

\subsection{Quantum Key Distribution protocols}

For an implementation of the BB84 or six-state protocols, the only required resources are the preparation, transmission and measurement of single-qubit states. That is, these protocols can be implemented in a prepare-and-measure set-up without the need of entanglement. However, both protocols have an equivalent entanglement-based implementation \cite{E91,BBM92}, in which a source (that may be in Alice's lab) produces a state that is distributed to Alice and Bob (ideally the maximally entangled state, $\ket{\Phi^+}=\frac{1}{\sqrt{2}}(\ket{00}+\ket{11})$). Then, upon receiving their systems, Alice and Bob perform measurements in randomly chosen bases (having two choices of basis for the BB84, and three for the six-state protocol). As long as the measurement devices are well characterized and controlled, an entanglement-based implementation allows one to relax the need for a precise characterization of the state preparation.
The entanglement-based version of the BB84 and six-state protocols played a key role to formalize their security proofs~\cite{SP00,Lo01,RennerThesis}.
From now on, for the purpose of our analyses, we focus on the entanglement-based version of these protocols. However we remark that for all our results, there is an equivalent implementation that requires only single-qubit states preparation.

The BB84 and the six-state protocols can be described by four main steps:
\begin{enumerate}
\item \textbf{Distribution and measurements:} Alice uses the source to produce a two-qubit state. She keeps one qubit and sends the other to Bob using a quantum channel. Upon receiving the systems, Alice and Bob, each randomly chooses a basis and performs the corresponding measurement. They repeat this procedure $N$ times. With the outcomes of their measurements they establish a string of $N$ bits each.
\item \textbf{Sifting and parameter estimation:} Alice and Bob communicate the measurement bases and discard the rounds in which different bases were used. Moreover, they sacrifice $m$ bits in order to estimate their average correlation and decide whether to abort or proceed with the protocol. If the protocol does not abort, the remaining $n$ bits constitute the raw key.
\item \textbf{Advantage distillation (optional):} The goal is to classically post-process the raw key, in order to increase the correlation between Alice and Bob, and get an advantage over the eavesdropper.
\item \textbf{Information reconciliation and privacy amplification:} In this step, Alice and Bob first implement an information reconciliation protocol that 
allows Bob to correct his bit-string for errors. Finally, they apply a privacy amplification protocol to transform a partially secure key of $n$ bits into a secure key of $l<n$ bits. 
\end{enumerate}

The parameters estimated in Step 2 are determined by the particular protocol. In an implementation of the BB84 protocol, Alice and Bob can estimate the values of the quantum bit error rates (QBERs) in the $X$ and $Z$ bases, $\qx$ and $\qz$. For the six-state protocol, they will also have an estimate of the QBER in the $Y$-basis, $\qy$. The QBER in the $Z$($X$)-basis, $Q_Z(Q_X)$, is the probability that Alice and Bob get different outcomes when they both measure their systems in the basis $Z$($X$). The QBER in the $Y$-basis is defined in a similar way, however, since the target state $\ket{\Phi^+}$ exhibits anti-correlation in the $Y$-basis, Bob flips his outcomes whenever he chooses to measure in the $Y$-basis.

In the originally proposed BB84 and six-state protocols, all the bases were chosen with equal probability among the set of bases specified by the protocols.
However, as shown in \cite{LCA05}, the efficiency of these protocols can be increased, without compromising security, if one of the bases is chosen with a higher probability. Then, in the asymptotic limit, the preferred basis is used almost all the time. The remaining bases are used only occasionally in order to test for the eavesdropper. This significantly increases the key rates, as only a small fraction of the rounds are discarded in the sifting process. In these protocols, the raw key is usually created from the rounds in which the preferred basis is used. The remaining rounds, in which other bases were chosen, are used for parameter estimation. For this reason, we denote the basis chosen with higher probability as the \emph{key generation basis}.

Advantage distillation, in Step 3, is an optional step. It consists of  Alice and Bob using two-way classical communication to select parts of the raw key that exhibit stronger correlation. This method was introduced in the context of classical protocols \cite{Mau93} and was proven to
be useful for 
 quantum protocols as well \cite{GL03,KBR07}. Usually, advantage distillation leads to significant drops in the key rate for the low noise regime, but it can considerably increase the noise tolerance of a protocol. For example, 
a BB84 implementation subjected to depolarizing noise (in which $\qx=\qz$) without advantage distillation can tolerate up to $11\%$ of QBER. If some advantage distillation is performed, the noise tolerance can be increased to $20\%$ QBER~\cite{GL03}.
 Advantage distillation protocols that have better performance in the low noise regime were also proposed~\cite{Watanabe}. 
 
 {Information reconciliation, in Step 4, aims at correcting Bob's string in order to make it equal to Alice's string. Information reconciliation can be implemented using only one-way communication from Alice to Bob.  Interactive protocols~\cite{BS94}, which are efficient to implement, are broadly used in QKD implementations~\cite{cascadeApp}. These protocols require two-way communication, however, they should not be confused with advantage distillation performed in Step 3. In advantage distillation, two-way communication is essential and, moreover, both Alice's and Bob's strings are modified during the protocol.}

\section{Results}\label{sec:results}
\subsection{Key rates for protocols without advantage distillation}\label{subsec:1way}

We first consider the BB84 and the six-state protocols when Alice and Bob skip Step 3. After measuring their quantum systems, Alice and Bob proceed to perform information reconciliation and privacy amplification. Information reconciliation protocols based on two-universal hashing functions leak the minimum amount of information necessary to correct for errors in Bob's string~\cite{BS94}. In~\cite{RW05}, it was proven that the minimum leakage of a one-way information reconciliation protocol is given by $n H(A|B)+\mathcal{O}(\sqrt{n})$, where $H(A|B)$ is the entropy of Alice's output conditioned on Bob's output, defined as $H(A|B)=-\sum_{a,b}p(A=a,B=b)\log p(A=a|B=b)$ with $p(A=a,B=b)$ being the probability that Alice and Bob obtain outcomes $a$ and $b$, respectively, for the measurement in consideration, $p(A=a|B=b)=\frac{p(A=a,B=b)}{p(B=b)}$ is the conditional probability, and the logarithms are taken in basis 2.

In order to analyze the key rate of the BB84 and the six-state protocol we can assume w.l.o.g. that the state distributed by the source is a Bell-diagonal state~\cite{RennerThesis}:
\begin{align}\label{eq:Belldiagstate}
\rho=\lambda_{00}\Phi_{00}+\lambda_{01}\Phi_{01}+\lambda_{10}\Phi_{10}+\lambda_{11}\Phi_{11},
\end{align}
where $\Phi_{ij}=\ketbra{\Phi_{ij}}{\Phi_{ij}}$ and $\ket{\Phi_{ij}}=X^iZ^j\otimes I \ket{\Phi^+}$. 
{Restricting the analysis to Bell-diagonal states is sufficiently general 
because for all states $\rho'$ such that 
$\lambda_{ij}=\bra{\Phi_{ij}}\rho'\ket{\Phi_{ij}}$, the corresponding Bell-diagonal state exhibits the same QBERs as the original state and leads to the lowest key rate~\cite{RennerThesis}.}

For the security analysis, it is also assumed that the measurements are perfect and all the noise can be mapped into the distributed state. In this case, the Bell coefficients $\DE{\lambda_{ij}}$ relate to the QBERs $\qx$, $\qy$, and $\qz$ by the following:
\begin{align}\label{eq:BelldiagstateQBER}
\begin{split}
\lambda_{00}&=1-\frac{(\qx+\qy+\qz)}{2},\\
\lambda_{01}&=\frac{\qx+\qy-\qz}{2},\\
\lambda_{10}&=\frac{-\qx+\qy+\qz}{2},\\
\lambda_{11}&=\frac{\qx-\qy+\qz}{2}.
\end{split}
\end{align}
Note that a Bell-diagonal state is completely characterized by the three QBERs $(\qx,\qy,\qz)$.

The key rates for the BB84 and the six-state protocols can be determined as a function of the coefficients of the estimated Bell-diagonal state, and therefore as a function of the QBERs. In a real implementation, in which a  finite number of rounds is considered, statistical effects play a role in the value of the key rate. However, for a sufficiently large number of rounds, the key rate of the six-state protocol, with $Z$ being the key generation basis, approaches the asymptotic value given by~\cite{Lo01,KGN05,RennerThesis}:
\begin{align}\label{eq:key6state}
R_{\rm 6state}=1-H(\DE{\lambda_{ij}}),
\end{align}
where $H(\DE{\lambda_{ij}})=\sum_{ij} -\lambda_{ij}\log \lambda_{ij}$, and the logarithms are taken in basis 2.

For BB84, since information about $Q_Y$ is not available, the key rate is given by the minimum over all possible values of $Q_Y$. This results in \cite{SP00,KGN05,RennerThesis}:
\begin{align}\label{eq:keyBB84}
R_{\rm BB84}=1-h(Q_X)-h(Q_Z),
\end{align}
where $h(x)=-x \log x -(1-x)\log(1-x)$ is the binary entropy.

If one of the other available bases is used for the key generation, the corresponding key rate can be obtained by simply permuting the QBERs in expressions \eqref{eq:key6state} and \eqref{eq:keyBB84}.
{Note that choosing the basis for key generation implies a choice of protocol. Once the key generation basis is fixed, Alice and Bob can run the protocol using this basis for almost all the rounds. The key rate now depends on the 
 information available to the eavesdropper when the estimated state, given in eq.~\eqref{eq:Belldiagstate}, is measured in the chosen key generation basis. 
A measurement in the  $X$-basis, or the $Y$-basis, performed on the estimated state, eq.~\eqref{eq:Belldiagstate}, can be seen as a $Z$-basis measurement performed on a rotated state.
The corresponding rotated state for when the $X$-basis, or the $Y$-basis, is used for the key generation measurement relates to eq.~\eqref{eq:Belldiagstate} by a permutation of the coefficients $\DE{\lambda_{ij}}$.} Eqs.~\eqref{eq:key6state} and \eqref{eq:keyBB84} are invariant under permutation of the QBERs. Therefore, the resulting key rate for protocols that use information reconciliation with minimum leakage is the same regardless of which of the available bases, {$\DE{X,Z}$ for BB84 and $\DE{X,Y,Z}$ for six-state}, is chosen for the key generation rounds.
This is stated in Proposition~\ref{prop:keyrates}.

\begin{prop}\label{prop:keyrates}
In an implementation of the BB84 or the six-state protocol in which an information reconciliation protocol with minimum leakage is used, the asymptotic key rate does not depend on which of the available bases is chosen as the key generation basis.
\end{prop}

\noindent \emph{Remark:} It is important to remark that Proposition \ref{prop:keyrates} takes into account that information reconciliation is performed using a protocol with minimum leakage. The minimum leakage is asymptotically given by $h(Q_M)$, where $M$ is the basis used for the key generation rounds, $M\in \DE{X,Z}$ for BB84 and $M\in\DE{X,Y,Z}$ for the six-state protocol. Information reconciliation protocols with minimum leakage cannot be implemented in practice,
and protocols that have higher leakage are used instead. There exist efficient information reconciliation protocols with asymptotic leakage given by $f\cdot h(Q_M)$, where $ f \leq 1.2$~\cite{vDK97,ELAB09,TMPE17}. The use of a sub-optimal information reconciliation protocol creates an asymmetry of the QBERs in the key rate, and in this case, in order to maximize the key, it is advantageous to choose the basis with the lowest QBER for the key generation rounds.\vspace{1em}

Now we state an interesting fact regarding the key rates of the six-states protocol when an information reconciliation protocol with minimum leakage is used.

\begin{fact}\label{fact1way}
For fixed values of $Q_X$ and $Q_Z$, the key rate is \emph{not} a monotonically decreasing function of $\qy$.
\end{fact}

\begin{proof}
If we fix the values of $\qx$ and $\qz$, in order to ensure positivity of the corresponding Bell-diagonal state, the possible values of $\qy$ are in the range
\begin{align}\label{eq:Qy}
|\qx-\qz|\leq \qy \leq \qx+\qz.
\end{align}
Additionally, we require that $\qx+\qy+\qz<1$ in order to have an entangled state.
One can see this by inspecting eq.~\eqref{eq:BelldiagstateQBER}.

Now, evaluating the derivative of the key rate, eq.~\eqref{eq:key6state}, with respect to $\qy$, we conclude that the minimum occurs for
\begin{align}
\qy^*=\qx+\qz-2\qx\qz,
\end{align}
which can be strictly smaller than the maximum attainable value for $\qy$.
\end{proof}

Observation~\ref{fact1way} is illustrated in Figure~\ref{fig:Qz6state} for the family of Bell-diagonal states $(\qx=0.1,\qy,\qz=0.1)$.

\begin{figure}[H]
\begin{center}
  \includegraphics[scale=0.60]{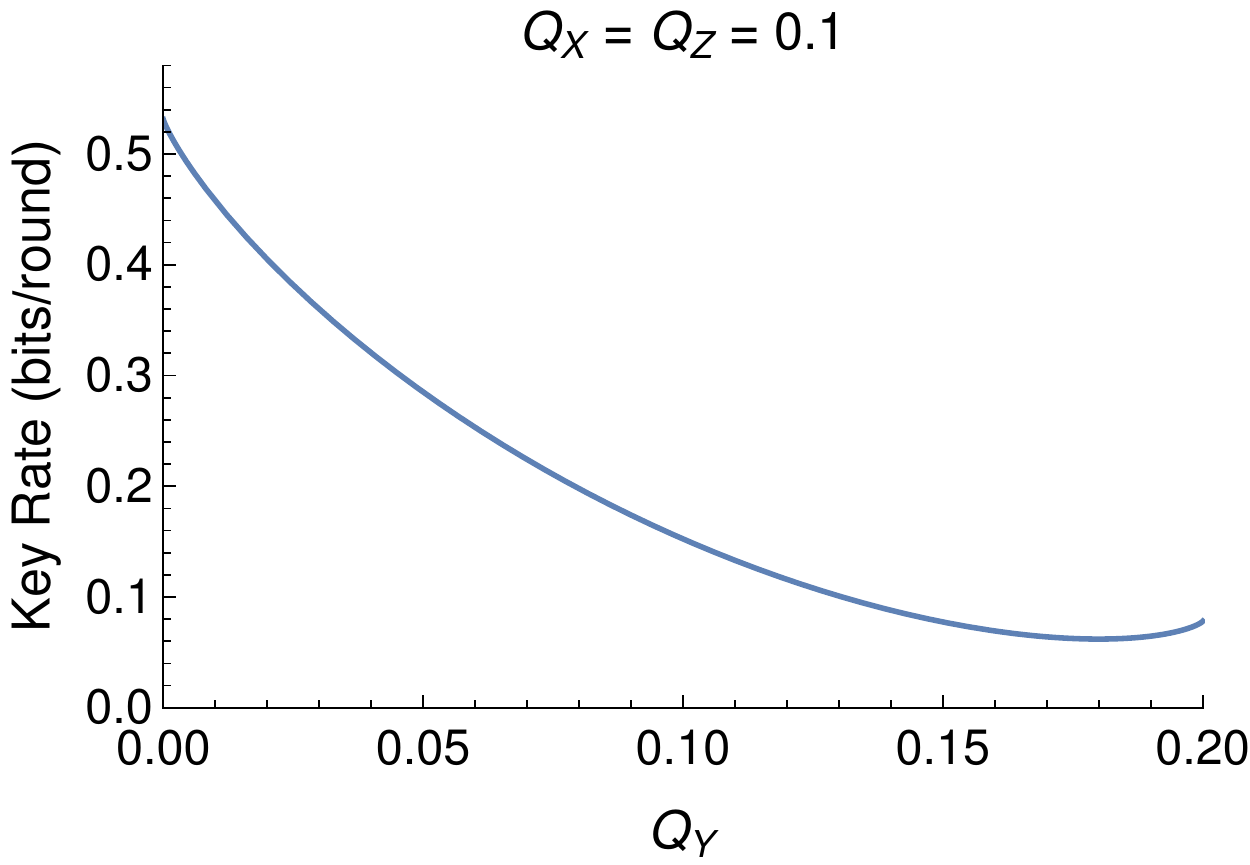}
 \caption{Asymptotic key rate of the six-state protocol with minimum leakage information reconciliation as a function of the QBER in the $Y$-basis, $Q_Y$, for the family of Bell-diagonal states with $Q_X=Q_Z=0.1$.}\label{fig:Qz6state}
 \end{center}
\end{figure}

Note that the minimum of the curve in Figure~\ref{fig:Qz6state} gives the key rate for the BB84 protocol when $\qx=\qz=0.1$. 

Observation~\ref{fact1way} together with the continuity of the key rate for the six-state protocol implies the following corollary.

\begin{cor}\label{cor:QBERs} There exist a state $\rho^{(1)}$ with QBERs $(\qx^{(1)},\qy^{(1)},\qz^{(1)})$ and a state $\rho^{(2)}$ with QBERs $(\qx^{(2)},\qy^{(2)},\qz^{(2)})$ such that 
\begin{align*}
\qx^{(1)}>\qx^{(2)}\;,\;\qy^{(1)}> \qy^{(2)}\;,\;\qz^{(1)}>\qz^{(2)}
\end{align*}
and 
\begin{align*}
 R_{\rm 6state}(\rho^{(1)})>R_{\rm 6state}(\rho^{(2)}).
\end{align*}
\end{cor}

As an example of Corollary \ref{cor:QBERs}, take $\rho^{(1)}$ to be the state with QBERs $\qx^{(1)}=\qz^{(1)}=10\%$ and $\qy^{(1)}=20\%$, and $\rho^{(2)}$ to be the state with $\qx^{(2)}=\qz^{(2)}=9.8\%$ and $\qy^{(2)}=18\%$. It holds that $R_{\rm 6state}(\rho^{(1)})>R_{\rm 6state}(\rho^{(2)})$.\vspace{1em}

We now investigate the behavior of the singlet fidelity and the entanglement of formation \cite{BDSW96} for the family of states considered in Figure~\ref{fig:Qz6state}.

For an entangled Bell-diagonal state $\rho$, eq.~\eqref{eq:Belldiagstate}, with $\lambda_{00}>\frac{1}{2}$, the entanglement of formation~\cite{BDSW96}, $EoF$, is given by
\begin{align}\label{EoF}
 EoF(\rho)&=h\de{\frac{1}{2}+\sqrt{\lambda_{00}(1-\lambda_{00})}},
 \end{align}
and the singlet fidelity $F$ is
\begin{align}\label{Fidelity}
 F(\rho)&=\lambda_{00}.
\end{align}

Figure~\ref{figmonotones}  illustrates that both quantities are monotonically decreasing functions of $Q_Y$.
This supports the intuition that a state with higher QBER is less close to the ideal state. However, as stated in Observation~\ref{fact1way} - see also Figure~\ref{fig:Qz6state}, this monotonic behavior is not always observed in the key rates of the six-state protocol.

\begin{figure}[H]
\begin{subfigure}{0.24\textwidth}	
	\includegraphics[scale=0.34]{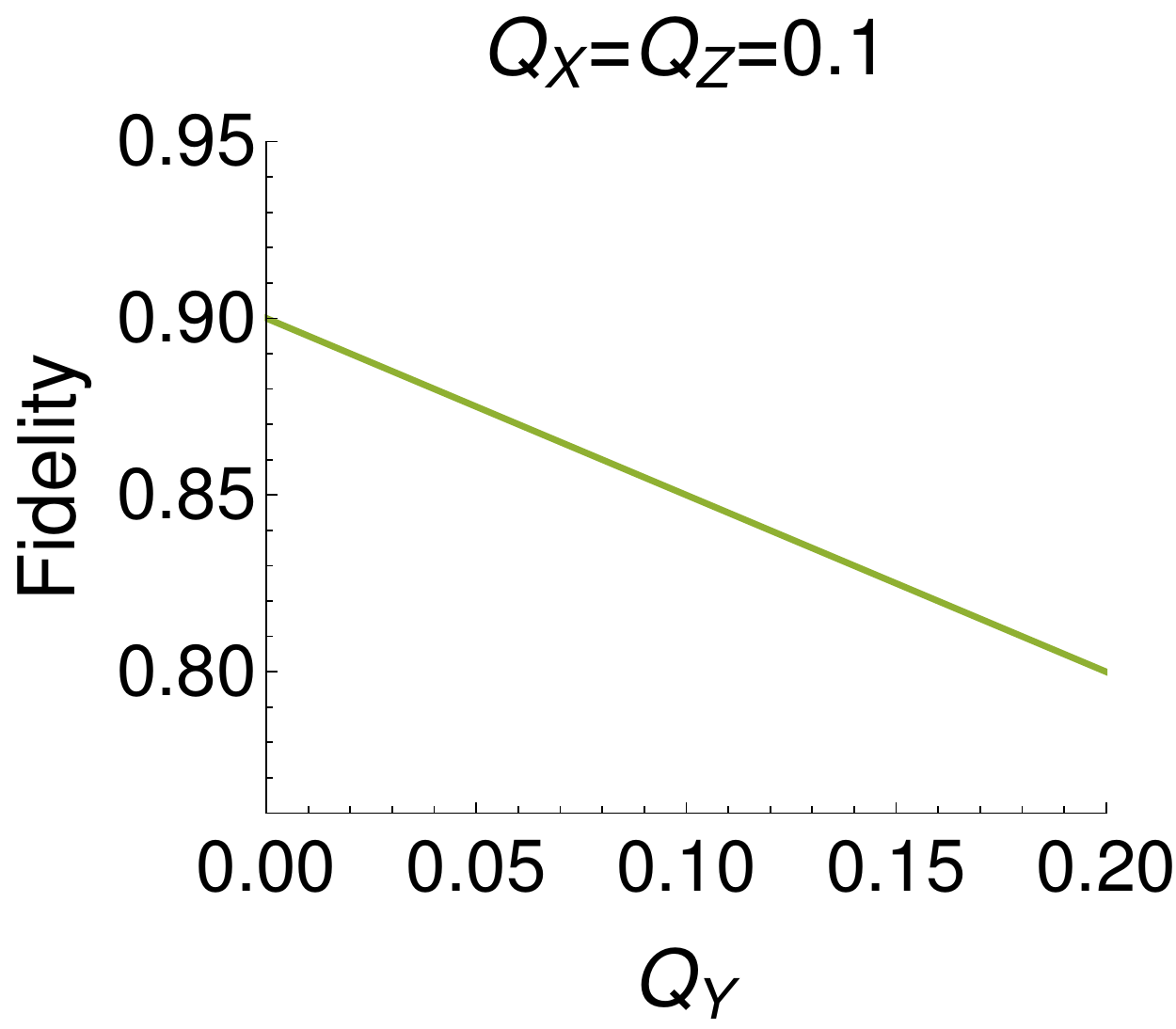}
\caption{}
\end{subfigure}%
\begin{subfigure}{0.25\textwidth}	
 	\includegraphics[scale=0.34]{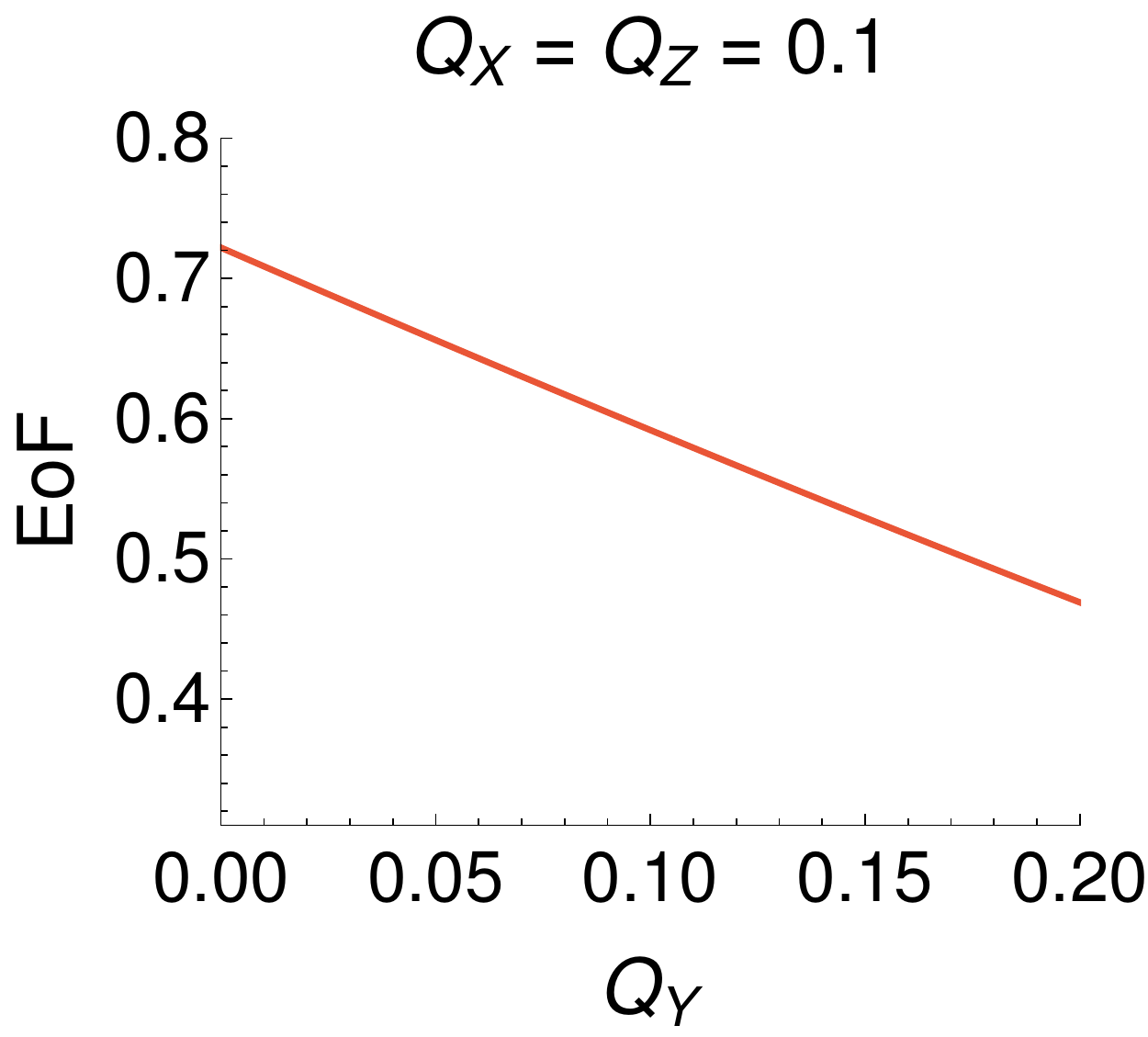}
 	\caption{ }
 	\end{subfigure}
\caption{Singlet fidelity (a) 
and entanglement of formation (b) for the family of Bell-diagonal state specified by QBERs $(0.1,\qy,0.1)$. Both quantities decrease monotonically as $\qy$ increases, in the range of possible values of $\qy$, $0\leq \qy \leq 0.2$. }\label{figmonotones}
 	\end{figure}

\subsection{Key rates for protocols with advantage distillation}\label{subsec:2way}

We now consider the BB84 and the six-state protocols when  advantage distillation is employed in Step 3. In particular, we consider the following advantage distillation protocol~\cite{BA07,KBR07,RennerThesis}:

\begin{algorithm}[H]{Protocol I: Advantage Distillation}\label{prot:AD}
\begin{algorithmic}[1]
\Statex Let $\DE{a_1,\ldots,a_n}$ and  $\DE{b_1,\ldots,b_n}$ be strings of bits held by Alice and Bob, respectively.
\State Alice and Bob divide their strings in blocks of two {consecutive} bits.
\For {each block $j$ of size 2}
\State Alice chooses a random bit $r\in \DE{0,1}$ and publicly communicates $(c_{j_1},c_{j_2}):=(a_{j_1}\oplus r, a_{j_2}\oplus r)$ to Bob.
\State Bob checks whether $(b_{j_1} \oplus c_{j_1}, b_{j_2}\oplus c_{j_2})\in \DE{\vec{0},\vec{1}}$. If that is the case he accepts, $acc=1$, else he sets $acc=0$.
\State Bob communicates $acc$ to Alice.
\State If $acc=1$ Alice keeps $a_{j_1}$ and Bob keeps $b_{j_1}$ for their raw key. Else they discard the two bits of the block.
\EndFor
\end{algorithmic}
\end{algorithm}

The key rates for BB84 and six-state protocol with advantage distillation, Protocol I, were derived in \cite{KBR07,RennerThesis}. For the six-state protocol, the key rate is given by 
\begin{align}\label{eq:key6stateAD}
R_{\rm 6state}^{{AD}}=\frac{1}{2}p_{\rm succ}^{AD}(1-H(\{\tilde{\lambda}_{ij}\})),
\end{align}
where $p_{\rm succ}^{AD}$ is the probability that Protocol I succeeds, \textit{i.e} that Alice and Bob do not discard a block. This occurs if either the two bits of Alice and Bob are equal or if both bits in the block are flipped. Note that steps 3 and 4 of Protocol I check whether the pair of bits of Alice and the pair of bits of Bob have the same parity.
If the raw key is generated by measurements in the 
$Z$-basis, then
\begin{align}\label{eq:psucc}
p_{\rm succ}^{AD}&=(\lambda_{00}+\lambda_{01})^2+(\lambda_{10}+\lambda_{11})^2\\
&=\qz^2+(1-\qz)^2.\nonumber
\end{align} 
And the coefficients $\DE{\tilde{\lambda}_{ij}}$ are given by
\begin{align}
\begin{split}
\tilde{\lambda}_{00}&=\frac{(\lambda_{00}+\lambda_{01})^2+(\lambda_{00}-\lambda_{01})^2}{2 p_{\rm succ}^{AD}},\\
\tilde{\lambda}_{01}&=\frac{(\lambda_{00}+\lambda_{01})^2-(\lambda_{00}-\lambda_{01})^2}{2 p_{\rm succ}^{AD}},\\
\tilde{\lambda}_{10}&=\frac{(\lambda_{10}+\lambda_{11})^2+(\lambda_{10}-\lambda_{11})^2}{2 p_{\rm succ}^{AD}},\\
\tilde{\lambda}_{11}&=\frac{(\lambda_{10}+\lambda_{11})^2-(\lambda_{10}-\lambda_{11})^2}{2 p_{\rm succ}^{AD}}.
\end{split}
\end{align}
where $\DE{\lambda_{ij}}$ relate to the QBERs by eq.~\eqref{eq:BelldiagstateQBER}.

In \cite{KBR07,RennerThesis} it was shown that applying advantage distillation, Protocol~I, has the same effect as if Alice and Bob would apply a quantum operation that brings two copies of a Bell-diagonal state with coefficients $\DE{\lambda_{ij}}$ into one copy of a Bell-diagonal state with coefficients $\DE{\tilde{\lambda}_{ij}}$ and then perform the measurement in this final state. This operation succeeds with probability $p_{\rm succ}^{AD}$. We will see, in Section~\ref{sec:entangdistil}, that the corresponding quantum operation is the application of bi-local CNOT gates {(i.e. Alice applies a CNOT to her two subsystems, and Bob does the same to his subsystems)}, followed by measurements of the target qubits and post-selection of the results.

The key rate for the BB84 protocol is obtained by taking the minimum of eq.~\eqref{eq:key6stateAD} over all possible values of $\qy$.

{Note that, for protocols with advantage distillation, the key rate, eq.~\eqref{eq:key6stateAD}, is not symmetric over permutation of the QBERs. Therefore, choosing a different basis (among the set of available bases) for  key generation may lead to different key rates}. We now state a curious observation about QKD protocols in which advantage distillation, given by Protocol I, is performed.

\begin{fact}\label{factAD}
In an  implementation of the BB84 or the six-state protocol in which advantage distillation, given by Protocol I, is performed, higher key rates may be obtained if the basis with higher QBER is used for the key generation rounds.
\end{fact}

\begin{figure}[H]
 \includegraphics[scale=0.85]{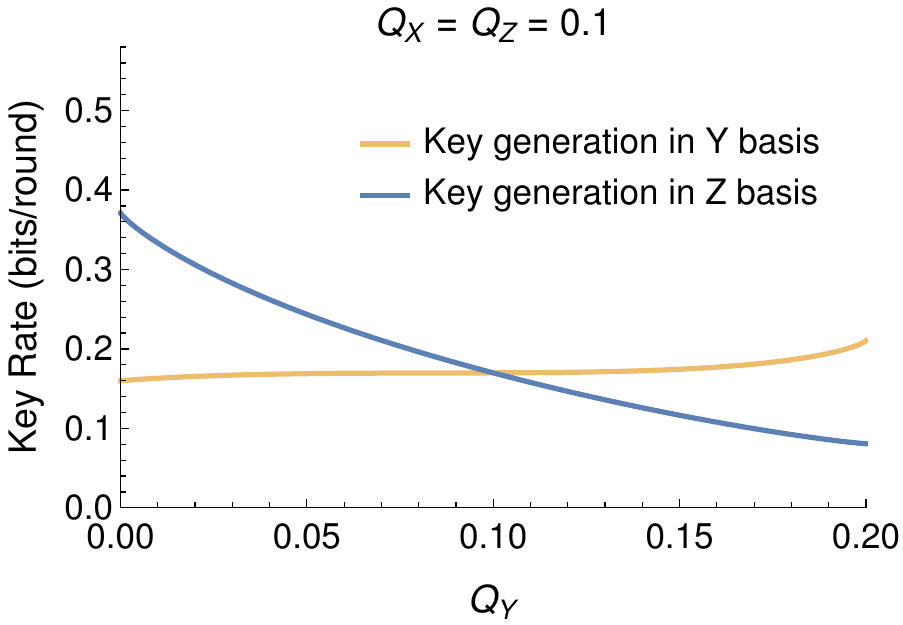}
 \caption{(Color online) Key rates for the six-state protocol with advantage distillation, given by Protocol I, for the family of states $Q_X=Q_Z=0.1$. Blue (dark gray) curve shows the key rate when the $Z$-basis is used for key generation, and the yellow (light gray) curve when the $Y$-basis is used for key generation.}\label{fig:Qz6stateAD}
\end{figure}

Figure~\ref{fig:Qz6stateAD} illustrates Observation~\ref{factAD} for the family of states $\qx=\qz=0.1$ considered in the previous section. In comparison with Figure~\ref{fig:Qz6state}, we note that, for lower values of $Q_Y$, higher rates are obtained when no advantage distillation is performed. However, as $Q_Y$ increases, an advantage is obtained with the use of advantage distillation and, specially, if the basis with higher QBER is used for the key generation rounds. Note also that, similar to Observation~\ref{fact1way}, an increase in the key rate with $Q_Y$ for high values of $\qy$ is also observed for the key generated with measurements in the $Y$-basis.

In order to explore Observation \ref{factAD} in more detail, we have performed an extensive numerical check over the range of possible values of QBERs $(\qx,\qy,\qz)$. 
For the BB84 protocol, Figure~\ref{fig:BB84AD3d} illustrates that, for almost all the  values of $\qx$ and $\qz$ that lead to positive key, the highest asymptotic key rate is obtained when the key generation basis is the one with higher QBER. This behavior inverts only for a small range of parameters, next to the limiting region where positive key can no longer be obtained. It is interesting to note that the success probability of the advantage distillation protocol, given in eq.~\eqref{eq:psucc}, is a monotonically decreasing function of the QBER of the key generation basis. However, even with the contribution of this factor to the key rate, see eq.~\eqref{eq:key6stateAD}, Figure~\ref{fig:BB84AD3d} shows that it is typically advantageous to use the basis with higher QBER for the key generation rounds.

\begin{figure}[H]
\centering
 \includegraphics[scale=0.9]{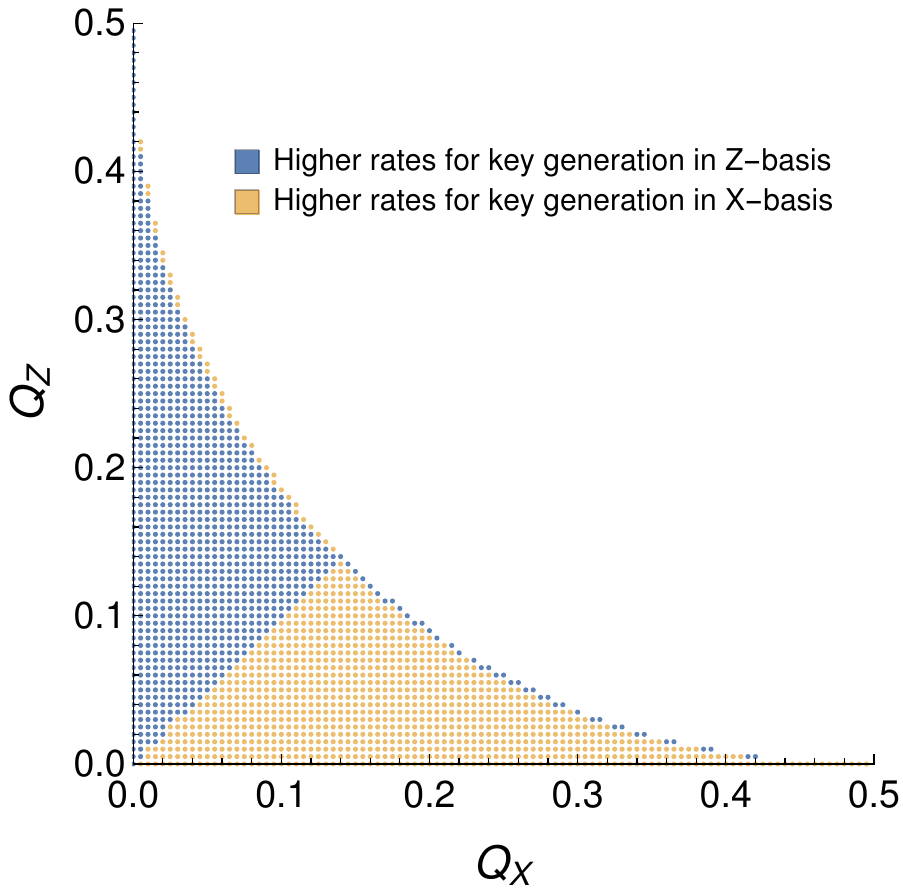}
 \caption{(Color online) Comparison of the key rates for the BB84 protocol with advantage distillation given by Protocol I when the $X$ and the $Z$ basis are used for the key generation. The blue (dark gray) dots  represent the parameters for which higher key rates are achieved when $Z$ is the key generation basis. The yellow (light gray) dots represent the case when higher rates are achieved if $X$ is used for key generation. {One can observe that, for almost all the range of parameters, higher rates are obtained when the basis with higher QBER is used for key generation. This behavior only inverts for a small range of parameters, near the limiting region where positive key can no longer be extracted.}}\label{fig:BB84AD3d}
\end{figure}

For the six-state protocol, we numerically compared, for the range of allowed parameters, the rates achieved when each of the three bases is used for key generation. Similarly to the BB84 case, we observed that higher key rates are obtained for key generation in the basis with higher QBER except for a small range of parameters. We found that it is not advantageous to use the basis with higher QBER for key generation only for some range of QBERs $(\qx,\qy,\qz)$ next to the region where no key can be obtained. As an example, for a state with QBERs $(\qx=0.39,\qy=0.39,\qz=0.01)$, one can obtain a secret key only if the $Z$-basis is used for key generation in the six-state protocol with advantage distillation given by Protocol I.

It is interesting to remark that the advantage of using the basis with higher QBER for protocols with advantage distillation can also occur in practical implementations where information reconciliation is performed using an one-way protocol~\cite{vDK97,ELAB09,TMPE17} {with non-optimal asymptotic leakage $f\cdot h(Q_M)$, for $f\leq1.2$}. Indeed, the advantage obtained by the use of a basis with higher QBER can be sufficiently large to compensate for the penalty of using an efficient information reconciliation protocol with higher leakage.

Protocol I can be generalized to blocks of arbitrary size $b$ \cite{Mau93}. Even though this leads to a significant decrease in the key rate in the low noise regime, higher noise tolerance can be achieved \cite{GL03,KBR07,BA07}. 
Figure \ref{fig:BB84AD7} illustrates the behavior of the key rates for the BB84 with advantage distillation using blocks of size 7. We now find that it is advantageous to use the basis with higher QBER for key generation in all the range of parameters that lead to positive key rate.

\begin{figure}[H]
\centering
 \includegraphics[scale=0.9]{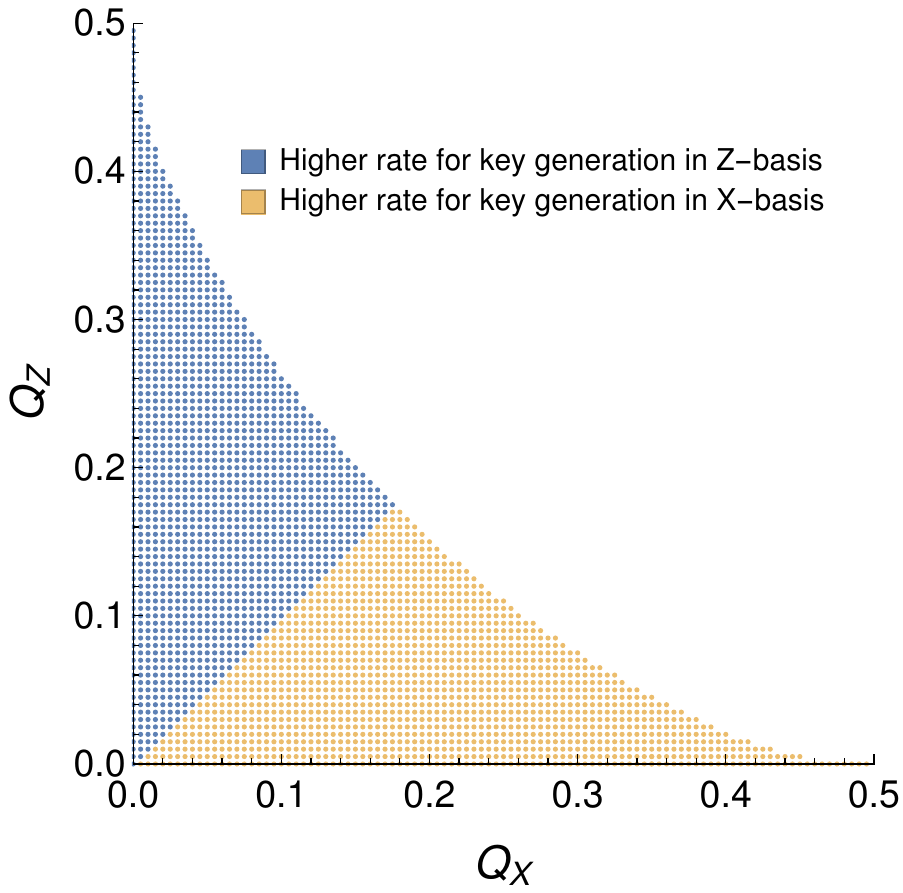}
 \caption{(Color online) Comparison of the key rates for the BB84 protocol with an advantage distillation protocol that uses blocks of size 7. The blue (dark gray) dots represent the parameters for which higher key rates are achieved when the $Z$-basis is used for key generation. The yellow (light gray) dots represent the case when higher key rates are achieved with measurements in the $X$-basis. {For this case, it is advantageous to always use the basis with higher QBER for key generation.}}\label{fig:BB84AD7}
\end{figure}

In Ref.~\cite{Watanabe}, Watanabe \textit{et al.} introduced an advantage distillation protocol that does not suffer from a big drop of the key rate in the low noise regime. The protocol introduced in  \cite{Watanabe} contains Protocol I as a sub-routine, and in the high noise regime the key rate coincides with the one obtained using Protocol I. Therefore, we expected that Observation~\ref{factAD} may also have an impact in this protocol. Indeed, in Figure \ref{fig:WatanabeAD} we illustrate that choosing the basis with higher QBER for key generation leads to higher rates for the family of states $Q_X=Q_Z=0.1$ when the advantage distillation protocol of Ref.~\cite{Watanabe} is performed. This effect played a role on the key rates estimated in Ref.~\cite{rozpkedek2018near}, for near-term implementations based on nitrogen-vacancy platforms and quantum repeaters.

\begin{figure}[H]
 \includegraphics[scale=0.9]{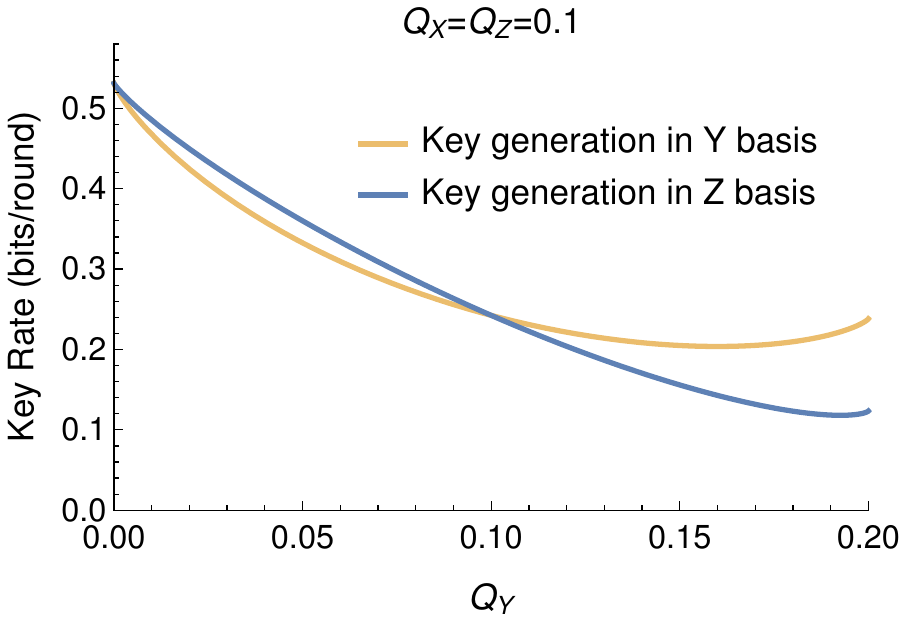}
 \caption{(Color online) Key rates for the six-state protocol with the advantage distillation protocol introduced in Ref.~\cite{Watanabe}, for the family of states $Q_X=Q_Z=0.1$. Blue (dark gray) curve illustrates the key rates when the $Z$-basis is used for the key generation rounds, and the yellow (light gray) curve when the $Y$-basis is used for the key generation rounds.}\label{fig:WatanabeAD}
\end{figure}

\section{QKD and entanglement distillation protocols}\label{sec:entangdistil}

In this section we show that the DEJMPS entanglement distillation protocol~\cite{DEJMPS} is related to advantage distillation, Protocol I, when it is applied to a string of bits generated by measurements in the basis with higher QBER.

A maximally entangled state provides a perfectly secure bit of key. Therefore, quantum key distribution and entanglement distillation are closely related~\cite{SP00,E91}. In fact, if the states shared by Alice and Bob have distillable entanglement, Alice and Bob could first distill maximally entangled states out of their noisy shared states, and then proceed to extract a perfectly secure key by measuring the distilled states. 
Interestingly, some entanglement distillation protocols can be completely mapped into a classical post-processing of the string of bits obtained after measurements on the initial states~\cite{GL03}. In that case, the entanglement distillation protocol has a corresponding QKD protocol that can be implemented in a prepare-and-measure set-up.
In a prepare-and-measure protocol, measurements of the quantum states are performed as soon as the states are received by Alice and Bob. This is of great interest for practical implementations, as no quantum memory is required to implement these protocols.
In \cite{GL03}, Gottesman and Lo characterized the properties that an entanglement distillation protocol needs to satisfy in order to be turned into a prepare-and-measure QKD protocol. The main idea is that some quantum operations (as CNOT gates) can be translated into classical operations (as XOR of the bits) on the string of bits generated by measurements in the initial state.

One-way information reconciliation based on hashing functions followed by privacy amplification is closely related to one-way entanglement distillation protocols based on Calderbank-Shor-Steane (CSS) codes that can correct for the corresponding amount of bit flip and phase flip errors \cite{SP00,Lo01}.
Similarly, the advantage distillation protocol, given by Protocol I, can be related to a two-to-one entanglement distillation protocol that takes two copies of a two-qubit state and maps it into one two-qubit state, hopefully more entangled than the original ones. {Under the assumption that the eavesdropper is restricted to individual attacks, it has been shown \cite{GW99,AMG03} that, a positive key can be extracted from any entangled state if advantage distillation with blocks of arbitrarily large size is applied. This shows that prepare-and-measure implementation is as powerful as entanglement distillation if the eavesdropper is restricted to individual attacks. This result was also generalized to high-dimensional QKD \cite{AGS03,Bruss03}. However, this equivalence does not hold true under general attacks~\cite{BA07}.}

{In the following, we will focus on a two-to-one entanglement distillation protocol. We will prove an interesting relation betwen the two-to-one entanglement distillation protocol  introduced in \cite{DEJMPS}, the DEJMPS protocol, and advantage distillation with blocks of size two, given by Protocol I.}

\begin{algorithm}[H]{Protocol II: DEJMPS entanglement distillation}\label{prot:AD}
\begin{algorithmic}[1]
\Statex Consider that Alice and Bob share $n$ copies of a Bell-diagonal state.
\State Alice and Bob {apply local unitary operations to each copy of their states in order to bring them to the form:}
\begin{align}\label{eq:stateDEJMPS}
\begin{split}
\rho=\lambda_{00} \Phi_{00} + \lambda_{10}\Phi_{10}  +\lambda_{11}\Phi_{11}+\lambda_{01}\Phi_{01},\\
\text{s.t. } \lambda_{00}> \frac{1}{2} \text{ and } \lambda_{00} >\lambda_{10}\geq \lambda_{11}\geq \lambda_{01}.
\end{split}
\end{align}
 \For {every $2$ systems}
\State Apply bi-local CNOT gates between the two copies.
\State Measure the target qubits and communicate the results.
\State If the measured flags are 00 or 11 keep the first system. Else discard both pairs.
\EndFor
\end{algorithmic}
\end{algorithm}

Protocol II includes a rotation of the initial states, step 1, before the application of the CNOT gates. 
{Any entangled two-qubit Bell-diagonal state can be brought to the form \eqref{eq:stateDEJMPS} by local unitaries (e.g., applying a Hadamard to each qubit leads to a permutation of the Bell states $\Phi_{01}$ and $\Phi_{10}$).}
The originally proposed DEJMPS protocol \cite{DEJMPS} includes  specific rotations that are independent of the input state.
In \cite{DNMV03}, it was proven that bringing the state to the form \eqref{eq:stateDEJMPS}, before applying the CNOT gates, maximizes the fidelity of the output state. Here we include in the DEJMPS protocol, as a first step, the rotations that optimize the output fidelity. Moreover, note that the DEJMPS protocol can also be applied to non-Bell-diagonal states, as any two-qubit entangled state that has singlet fidelity higher than $\frac{1}{2}$ can be brought to the form \eqref{eq:stateDEJMPS} by local operations and classical communication \cite[Appendix A]{BDSW96}. 

The following theorem states that the DEJMPS protocol is actually related to the six-state protocol with advantage distillation, given by Protocol I, when the basis with higher QBER is chosen for the key generation.

\begin{thm}\label{thm:AD_DEJMPS}
The following two procedures result in the same key rates:
\begin{itemize}
 \item[(i)] Alice and Bob implement the six-state protocol with advantage distillation, given by Protocol I, using the basis with higher QBER for key generation.
 \item[(ii)] Alice and Bob apply the DEJMPS protocol to every two copies of their states, and subsequently implement the six-state protocol, without advantage distillation, by measuring the distilled states.
\end{itemize}
\end{thm}

\begin{proof}
In order to prove the equivalence of procedures (i) and (ii), we show that generating a string of bits by performing measurements in the basis with higher QBER followed by advantage distillation is equivalent to performing the DEJMPS entanglement distillation protocol and measuring the resulting state in the $Z$-basis.

We first note that  steps 2-6 of Protocol II, followed by the measurement of the remaining qubits, can be, equivalently, implemented in a prepare-and-measure scenario. This is due to the following observations:
(a) the CNOT gate commutes with measurements in the $Z$-basis; (b) the final measurements performed on the remaining control qubits, also commute with the post-selection operation applied in step 5. This is due to the fact that in step 5, a pair of qubits is discarded according to the outputs of the target qubits only, therefore this operation acts as the identity on the control qubits of the remaining systems. Observations (a) and (b) imply that in an implementation of steps 2-6 of Protocol II followed by measurement of the remaining qubits, one can first measure all the subsystems and then proceed to apply the CNOT gate, step 3, and post-selection of results, step 5. {In this equivalent description, in which all the systems are measured first, the 
CNOT gate and the post-selection act on classical strings. Their action, then, corresponds to Alice and Bob locally computing the XOR of their respective two bits and comparing if they have the same parity}.
And this is exactly the action of the advantage distillation described in Protocol I {(note that, in Protocol I, $acc=1$ iff $a_{j_1}\oplus a_{j_2}=b_{j_1}\oplus b_{j_2}$)}.

Protocol II includes a local rotation of the quantum states in step 1. Instead of rotating the state, we could equivalently rotate the operations. From the previous paragraphs we have seen that procedure (ii) of the Theorem can be implemented by performing measurements in the $Z$-basis on all the subsystems before applying the CNOT gates and post-selection of the results. The first step of measurements is described as
\begin{widetext}
\begin{align}
 \de{\sum_i\ketbra{i}{i}_A\otimes \sum_i\ketbra{i}{i}_B }\cdot \rho_{AB}\cdot \de{\sum_i\ketbra{i}{i}_A\otimes \sum_i\ketbra{i}{i}_B}
\end{align}
\end{widetext}
If the initial states are rotated before the measurements, $\rho_{AB}\mapsto U_A\otimes U_B\cdot \rho_{AB} \cdot U^{\dag}_A\otimes U^{\dag}_B$, then we have
\begin{widetext}
\begin{align}
\begin{split}
&\de{\sum_i\ketbra{i}{i}_A\otimes \sum_i\ketbra{i}{i}_B}\cdot (U_A\otimes U_B\cdot\rho_{AB}\cdot U^{\dag}_A\otimes U^{\dag}_B) \cdot \de{\sum_i\ketbra{i}{i}_A\otimes \sum_i\ketbra{i}{i}_B}\\
&=\de{\sum_i\ketbra{i}{i}U_A\otimes \sum_i\ketbra{i}{i}U_B}\cdot \rho_{AB}\cdot \de{\sum_iU^{\dag}_A\ketbra{i}{i}\otimes \sum_i U^{\dag}_B\ketbra{i}{i}}.\label{eq:rotateMesure}
\end{split}
\end{align}
\end{widetext}
Since the rotations are local, they can be mapped into the measurements, and the last expression of \eqref{eq:rotateMesure} describes measurements in the rotated bases  $\DE{U^{\dag}_A\ket{i}}$ and $\DE{U^{\dag}_B\ket{i}}$ on the original state $\rho_{AB}$.
By the equivalence established in the previous paragraph, Protocol II, followed by measurement of the remaining qubits, corresponds to applying  the advantage distillation, Protocol I, to a string of outcomes obtained from measurements performed in the corresponding rotated basis.

Finally, let us evaluate the effect of the rotation performed in step 1 of Protocol II. The QBERs of a Bell-diagonal state are given by  
\begin{align}
\begin{split}
Q_X&=\lambda_{01}+\lambda_{11},\\
Q_Y&=\lambda_{01}+\lambda_{10},\\
Q_Z&=\lambda_{10}+\lambda_{11}.
\end{split}
\end{align}
Given the relations satisfied by the coefficients after the rotation in step 1, eq.~\eqref{eq:stateDEJMPS}, we have that $Q_Z\geq Q_X$ and $Q_Z \geq Q_Y$. Therefore, after the initial rotation, $Z$ is the basis with the highest QBER. In the alternative picture in which all the measurements are performed first, and the rotation of the state is mapped into a rotation of the  measurements, see eq.~\eqref{eq:rotateMesure}, we have that the rotated bases, $\DE{U^{\dag}_A\ket{i}}$ and $\DE{U^{\dag}_B\ket{i}}$, correspond to measurements in the basis with the highest QBER.

So the DEJMPS protocol followed by a measurement in the $Z$-basis, corresponds to advantage distillation, given by Protocol I, applied to the outcomes of measurements in the basis with highest QBER. This proves the equivalence of procedures (i) and (ii) of Theorem~\ref{thm:AD_DEJMPS}.
\end{proof}

Theorem~\ref{thm:AD_DEJMPS} establishes that the DEJMPS protocol falls into the category of entanglement distillation protocols that have a corresponding prepare-and-measure QKD as characterized by Gottesman and Lo~\cite{GL03}. {Moreover, it shows that the particular rotation introduced in step 1 of Protocol II can be implemented in the prepare-and-measure scenario by choosing the basis with higher QBER for key generation.}

In \cite{DNMV03}, it was shown that the DEJMPS protocol, Protocol II, is the two-to-one entanglement distillation protocol that achieves the highest fidelity, among all possible protocols that involve Pauli rotations. In \cite{RSPEDW18} it was proven 
that the DEJMPS is the optimal two-to-one entanglement distillation protocol for rank 3 Bell-diagonal states. \textit{I.e.}, the highest possible fidelity is achieved with the highest possible probability of success, considering all possible protocols that take two copies into one.  
We now state analogous results for the key rates of the corresponding QKD protocol.
 
A Bell-diagonal state of rank up to 3 satisfies that one of the QBERs is equal to the sum of the other two. Without loss of generality, we can consider $\qy=\qx+\qz$. The corresponding Bell-diagonal state is then
\begin{align}\label{eq:rank3}
\rho=(1-\qx-\qz) \Phi_{00} + \qz\Phi_{10} +\qx \Phi_{11} .
\end{align}
We numerically compared the key rates achieved by the state given in eq.~\eqref{eq:rank3} for the six-state protocol with advantage distillation given by Protocol I and key generation in all of the three bases. In the region of positive key rate, we observed that, over all the range of  values of $\qx$ and $\qz$, higher rate is achieved when $Y$ is the key generation basis. 

For rank 4 states, using the basis with highest QBER is not always advantageous. 
As mentioned in Section~\ref{sec:results}, 
 a counter-example is given by the state with QBERs $(Q_X=0.39,Q_Y=0.39,Q_Z=0.01)$. For a Bell-diagonal state with the specified QBERs, a positive key rate can only be obtained by performing measurements in the $Z$-basis in an implementation of the six-state protocol with advantage distillation given by Protocol I. We remark, however, that this does not contradict the fact that the corresponding state after an application of the DEJMPS procedure has higher fidelity. Indeed, as we have seen from Observation 1, higher fidelity does not necessarily imply higher key rates in the six-state protocol.
Analyzing this example in detail, we find that the fidelity of the initial state is  $0.605$ and no key can be extracted by directly applying information reconciliation and privacy amplification. If entanglement distillation is performed without the previous rotations, \textit{i.e.} by applying only steps 4-6 of Protocol II, the final fidelity is $0.525$ and positive key can be extracted from the corresponding final state using a six-state protocol with an optimal one-way hashing information reconciliation. Applying the DEJMPS protocol, in which the initial rotations are performed, we obtain a state with higher fidelity, equal to $0.698$, yet this state does not lead to positive key rate in the six-state protocol.

\section{Implications to experimental implementations}\label{sec:experiments}

We now discuss the implications of our results to {fibre-based} implementations of quantum key distribution over long distances.

The most common way of transmitting qubits over long distances is by using photons sent through optical fibres. One of the challenges of a fibre-based implementation is that the transmissivity of the channel decays exponentially with the distance. It has been shown that this also leads to an exponential decay of the achievable secret-key rate over such channel~\cite{takeoka2014fundamental, pirandola2015fundamental}, thus making practical QKD over direct fibre connections impossible for larger distances. Significant amount of both theoretical and experimental efforts are being invested into overcoming this problem using the so-called quantum repeaters~\cite{briegel1998quantum}, which have the capability of beating the exponential scaling of secret-key rate with distance. One of the fundamental building blocks of such quantum repeater schemes is a memory node that can store quantum information over time. By dividing the channel into elementary links, entanglement generation can be attempted independently over those segments, thanks to the quantum memories at the intermediate repeater stations. Unfortunately, quantum states stored in such memories decohere with time.

Decoherence is often a complex process that could be modeled by a composition of different noise channels depending on the physical implementation of the quantum memory. However, often the dominant type of noise corresponds to the dephasing channel. This has been observed for many physical platforms which are promising candidates for long-lived quantum memories, such as nitrogen-vacancy centres~\cite{reiserer2016robust,kalb2018,maurer2012room}, trapped ions~\cite{kielpinski2001entanglement} and neutral atoms~\cite{schrader2004neutral}. Therefore, the dephasing channel is frequently used to model memory decoherence in quantum repeater literature, thus leading to expected non-uniform QBER over the three bases~\cite{luong2015overcoming, parameterregimes,razavi2009quantum,nemoto2016photonic,rozpkedek2018near}. Hence, the results of this paper will be highly relevant for choosing the key generation basis for entanglement based QKD schemes implemented across a future quantum repeater network. In fact, some of the authors of this manuscript have already applied the results of this work into their model of near-term proof-of-principle quantum repeaters based on nitrogen-vacancy centres~\cite{rozpkedek2018near}. 

{Regarding prepare-and-measure QKD schemes, the method of decoy states was introduced to overcome vulnerabilities due to imperfections in the source~\cite{Hwang_Decoy,Lo_Decoy}. In a decoy state protocol, the asymptotic key rates, 
\eqref{eq:key6state} and \eqref{eq:keyBB84}, are modified to account for the information leakage from the rounds in which multiple photons are generated. The modified key rates have a much more intricate dependence on the QBERs and, therefore, a detailed analysis is required to determine the effects of asymmetric noise in decoy state implementations. We leave it as an interesting open question for future investigation.}

\section{Acknowledgements}
We thank K. Goodenough, J. Kaniewski, S. de Bone, and T. Coopmans for helpful discussions. This work was supported by the European Research Council through a Starting Grant, and the Netherlands Organisation for Scientific Research (NWO) through a VIDI grant, and a Zwaartekracht grant. GM was also funded by the Deutsche Forschungsgemeinschaft (DFG, German Research
Foundation) under Germany's Excellence Strategy – Cluster of Excellence
Matter and Light for Quantum Computing (ML4Q) EXC 2004/1 – 390534769. 
DE was also supported by the Netherlands Organization for Scientific Research
(NWO/OCW), as part of the Quantum Software Consortium program (project number
024.003.037/3368).

\bibliographystyle{unsrt}
\bibliography{biblioADvsQBER}

\end{document}